\documentclass[12pt]{amsart}

\date{January 22, 2013}

\usepackage{amsmath,amssymb,amsthm}
\usepackage{color}
\usepackage{hyperref}

\newtheorem{lem}{Lemma}[section]
\newtheorem{prop}{Proposition}[section]

\newtheorem{thm}{Theorem}[section]

\theoremstyle{definition}

\theoremstyle{remark}

\theoremstyle{remark}
\newtheorem{remark}{Remark}[section]
\numberwithin{equation}{section}

\newcommand{\R}{{\mathbb R}}

\definecolor{blu}{rgb}{0,0,1}

\begin{document}

\title[Existence of ground states at the binding threshold]{Existence of ground states for negative ions\\ at the binding threshold}

\author[J. Bellazzini]{Jacopo Bellazzini}
\address{J. Bellazzini, Universit\`a di Sassari, Via Piandanna 4, 07100 Sassari, Italy}
\email{jbellazzini@uniss.it}

\author[R. L. Frank]{Rupert L. Frank}
\address{R. L. Frank, Department of Mathematics, Princeton University, Princeton, NJ 08544, USA}
\email{rlfrank@math.princeton.edu}

\author[E. H. Lieb]{Elliott H. Lieb}
\address{E. H. Lieb, Departments of Mathematics and Physics, Princeton University, Princeton, NJ 08544, USA}
\email{lieb@princeton.edu}

\author[R. Seiringer]{Robert Seiringer}
\address{R. Seiringer, Department of Mathematics and Statistics, McGill University, 805 Sherbrooke Street West, Montreal, QC H3A 2K6, Canada} 
\email{robert.seiringer@mcgill.ca}

\thanks{\copyright 2013 by the authors. This paper may be reproduced, in its entirety, for non-commercial purposes.\\
Partial financial support from PRIN 2009 `Metodi Variazionali e Topologici nello Studio di Fenomeni non
Lineari' (J.B.), the U.S. National Science Foundation through grants PHY-1068285 (R.F.), PHY-0965859 (E.L.), the Simons Foundation (\# 230207, E.L.) and the NSERC (R.S.) is acknowledged.}

\begin{abstract}
As the nuclear charge $Z$ is continuously decreased an $N$-electron atom undergoes a binding-unbinding transition at some critical $Z_c$. We investigate whether the electrons remain bound when $Z=Z_c$ and whether the radius of the system stays finite as $Z_c$ is approached. Existence of a ground state at $Z_c$ is shown under the condition $Z_c<N-K$, where $K$ is the maximal number of electrons that can be removed at $Z_c$ without changing the ground state energy.
\end{abstract}

\maketitle

\section{Introduction and main result}

The energy of a quantum-mechanical system composed of $N$ electrons and one fixed nucleus of charge $Z>0$ is described by the Hamiltonian
\begin{equation}
\label{eq:hamintro}
\sum_{i=1}^N \left( p_i^2 -\frac{Z}{|x_i|} \right) +\sum_{i<j}\frac{1}{|x_i-x_j|} \,.
\end{equation}
(Here we use units in which the electron mass $m=1/2$, the electron charge $e=-1$ and Planck's constant $\hbar=1$.) The terms $p_i^2$, where $p_i=-{\rm i}\nabla_i$, and $-Z/|x_i|$ describe the kinetic energy of the $i$-th electron and its potential energy due to the attraction to the nucleus, respectively. The term $|x_i-x_j|^{-1}$ stands for the potential energy due to the repulsion between the $i$-th and the $j$-th electron. The Pauli principle dictates that the Hamiltonian is considered as acting in the subspace $L_a^2(\R^{3N})$ of \emph{anti-symmetric} functions in $L^2(\R^{3N})$, that is, $\psi$'s satisfying
$$
\psi(\ldots,x_i,\ldots,x_j,\ldots) = - \psi(\ldots,x_j,\ldots,x_i,\ldots)
\qquad\text{for}\ i\neq j \,. 
$$
(For the sake of simplicity, we ignore the electron spin. It can be included easily.)

The ground state energy of the system is given by the bottom of the spectrum of the Hamiltonian \eqref{eq:hamintro}. If this bottom of the spectrum is an eigenvalue, that is, if there is an eigenfunction in $L_a^2(\R^{3N})$, then the system is said to be bound and the eigenfunction describes its ground state.

It is intuitively clear that a nucleus of charge $Z$ can bind $N$ electrons if $Z$ is large compared to $N$, and that it cannot if $Z$ is small compared to $N$. This fact can also be shown mathematically: Zhislin \cite{Zh} proved that the system is bound if $Z> N-1$ and Nam \cite{Na} showed that the system is not bound if $N\geq 1.22 Z + 3Z^{1/3}$, improving the earlier condition $N\geq 2Z+1$ of \cite{Li}. For asymptotic results as $Z\to\infty$, see, for instance, \cite{Ru,Si,LiSiSiTh,SeSiSo}. In these results, and also in our paper, we shall consider $Z$ as an arbitrary positive (not necessarily integer) parameter. Then the cited results imply that for fixed $N$ there is (at least one) critical value of $Z$ where a binding-unbinding transition occurs.

In this paper we are interested in this binding-unbinding transition. More precisely, we investigate whether the system is bound at this critical $Z$ value, that is, whether a ground state exists. Intuitively, this question is related to the size of the system. As the continuous parameter $Z$ moves from the binding regime across the critical value into the unbinding regime, there are two possible scenarios: Either the size of the system increases indefinitely and becomes infinite as $Z$ reaches the critical value $Z_c$, or else the size of the system approaches a finite value at the critical value and then jumps discontinuously to infinity. The first scenario corresponds to the case where no ground state exists at the critical value, and the second one to where it does exist.

This question has been discussed in the physics literature (see, e.g., \cite{St}) and it was proved by Th. and M. Hoffmann-Ostenhof and Simon \cite{HOS} that for two-electron atoms in the spin singlet state the second scenario occurs, that is, there is a ground state at the critical coupling value. This corresponds to the case $N=2$, but \emph{without} the anti-symmetry assumption. On the other hand, Th. and M. Hoffmann-Ostenhof \cite{HO} showed that in the triplet $S$-sector the first scenario occurs. This corresponds to the $N=2$ case \emph{with} anti-symmetry, but the admissible functions are further restricted to depend only on $|x_1|$, $|x_2|$ and $x_1\cdot x_2$. Very recently Gridnev \cite{G} has, among other things, generalized the Hoffmann-Ostenhof--Simon existence result to arbitrary $N$, but with the additional assumption that $Z_c$ lies in the interval $(N-2,N-1)$. He has also conjectured that the assumption of a lower bound $N-2$ on the critical $Z$ is not necessary.

Our goal here in this paper is to reprove Gridnev's result by completely different means and to replace his assumption by a weaker one which we believe to be optimal. (In contrast to \cite{G}, however, we only consider an infinitely heavy nucleus.) Our method extends the one introduced in \cite{FLS}, where an alternative proof of the Hoffmann-Ostenhof--Simon existence result was given. In contrast to \cite{HOS} no positivity of the ground state is needed and the proof in \cite{FLS}, as well as the proof in this paper, extend to the case where the particles move in an external magnetic field, for instance.

We proceed to formulate our results precisely. It is physically equivalent and mathematically convenient to rescale the $x_i$ in \eqref{eq:hamintro} by $U=1/Z$. We then find that, apart from an overall factor of $Z^2$, the Hamiltonian \eqref{eq:hamintro} is unitarily equivalent to the Hamiltonian
\begin{equation}
H_U^{(N)}:=\sum_{i=1}^N \left( p_i^2 -\frac{1}{|x_i|} \right) +\sum_{i<j}\frac{U}{|x_i-x_j|} \,.
\end{equation}
We consider $H_U^{(N)}$ as a self-adjoint operator in the Hilbert space $L_a^2(\R^{3N})$ of anti-symmetric functions and denote its ground state energy by
$$
E_U^{(N)}= \text{ inf spec } H_U^{(N)}=\inf_{\|\psi\|=1} \langle \psi | H_U^{(N)} | \psi \rangle \,.
$$
The ground state energy $E_U^{(N)}$ of $H_U^{(N)}$ is a non-decreasing, concave function of $U$ and one has the ordering
$$
E_U^{(N)} \leq E_U^{(N-1)} \leq E_U^{(N-2)} \leq \ldots
$$
with respect to $N$. We define $\mathcal U^{(N)}=\{ U>0 :\ E_U^{(N)} < E_U^{(N-1)} \}$. The set of critical coupling constants is given by
\begin{align*}
\label{eq:critcoup}
\mathcal U^{(N)}_c & = \partial \mathcal U^{(N)} \\
& = \left\{ U>0 :\ E_U^{(N)} = E_U^{(N-1)}\ \text{and } E_{U_n}^{(N)} < E_{U_n}^{(N-1)} \text{ for some } U_n \to U \right\} \,.
\end{align*}
As mentioned before, by Zhislin's theorem $E_U^{(N)} < E_U^{(N-1)}$ for $U<1/(N-1)$ and, for instance by \cite{Li}, $E_U^{(N)} = E_U^{(N-1)}$ for all large $U$. Thus, $\mathcal U^{(N)}_c$ is non-empty. It is natural to believe that $\mathcal U^{(N)}$ is a single interval and that $\mathcal U^{(N)}_c$ contains only one element, but we do not know how to prove this.

According to the HVZ theorem (see, e.g., \cite{CyFrKiSi}) the strict inequality $E_U^{(N)} < E_U^{(N-1)}$ implies that $E_U^{(N)}$ is an eigenvalue of $H^{(N)}_U$ and, consequently, a ground state exists. On the other hand, it is clear that if $U\in\R_+\setminus\left( \mathcal U^{(N)} \cup \mathcal U^{(N)}_c \right)$, then $E_U^{(N)}$ is \emph{not} an eigenvalue of $H^{(N)}_U$.

The following theorem is our main result. It gives a sufficient condition for $E_{U_c}^{(N)}$ to be an eigenvalue of $H^{(N)}_{U_c}$ for $U_c\in\mathcal U^{(N)}_c$. The sufficient condition (for fixed $U_c\in\mathcal U^{(N)}_c$) depends on an integer $1\leq K\leq N-1$, which is the largest integer such that $E_{U_c}^{(N)}=E_{U_c}^{(N-K)}$. Thus, we have
$$
E_{U_c}^{(N)} = E_{U_c}^{(N-1)} =\ldots = E_{U_c}^{(N-K)} < E_{U_c}^{(N-K-1)} \,,
$$
where we interpret $E_{U_c}^{(0)}=0$. Physically, $K$ denotes the maximal number of electrons that can be removed from the system at $U=U_c$ without changing the energy.

Our main result is as follows.

\begin{thm}[Binding for $N$ electrons at threshold]\label{binding}
Let $U_c\in\mathcal U^{(N)}_c$ and assume that $U_c>\frac{1}{N-K}$, where $K$ is the largest integer such that $E_{U_c}^{(N)}=E_{U_c}^{(N-K)}$. Then $H_{U_c}^{(N)}$ has a ground state eigenfunction $0\not\equiv\psi_{U_c} \in L_a^2(\R^{3N})$.
\end{thm}

\begin{remark}
The theorem is also valid if the anti-symmetry assumption is dropped or, which is the same, replaced by a symmetry assumption. The proof in this unconstrained case follows along the same line and is actually somewhat simpler, see Remark \ref{symm}. The method also goes through essentially without modification in the case of spin.
\end{remark}

\begin{remark}
In the case $N=2$ and without the anti-symmetry assumption, the strict inequality $U_c>1$ is a classical result of Bethe \cite{Be}. Thus, (the symmetric version of) Theorem \ref{binding} extends the results of \cite{HOS} and \cite{FLS}. 
\end{remark}

\begin{remark}
In the case $N=2$ and with the anti-symmetry assumption enforced, one has equality $U_c=1$ and no ground state exists (at least no ground state depending only on $|x_1|$, $|x_2|$ and $x_1\cdot x_2$) \cite{HO}. This suggests that, in general, the assumption $U_c>\frac{1}{N-K}$ may not be dropped.
\end{remark}

\begin{remark}
Our theorem generalizes Gridnev's result (in the case of infinite nuclear mass). Indeed, Gridnev's assumption $U_c<\frac1{N-2}$ implies, by Zhislin's theorem, that $E_{U_c}^{(N-1)}< E_{U_c}^{(N-2)}$ and thus $K=1$. Gridnev's second assumption $U_c>\frac1{N-1}$ coincides with our assumption. We emphasize that the main difficulty that we overcome in this paper is the case $U_c\geq \frac1{N-2}$ or, more generally, the case where $K\geq 2$.
\end{remark}


\section{Strategy of the proof}

A standard way to prove the existence of ground states at threshold is to  prove that the weak limit of a ground state sequence $\psi_{U_n}$ when $U_n \rightarrow U_c$ is not zero. Once one knows that the weak limit is non-zero it is easy to see that this weak limit has to be a ground state. In the spirit of \cite{FLS} we give an upper bound on the `radius' of $\psi_{U_n}$ which is independent on $U_n$. Hereafter we denote $|x|_{\infty}:=\max\{|x_1|,..., |x_N|\}$. We interpret this as the maximal distance of the electrons from the nucleus.

\begin{thm}[Uniform upper bound on the radius of the atom]\label{main}
For any $\delta,\theta>0$ and any $1\leq K\leq N-1$ there are constants $m,R>0$ such that for all $U\in\mathcal U^{(N)}$ satisfying
\begin{equation}
\label{eq:mainass}
U \geq \frac{1}{N-K}+\delta
\quad\text{and}\quad
E_U^{(N-1)} \leq E_U^{(N-K-1)} -\theta
\end{equation}
and for all normalized ground states $\psi_{U}$ of $H^{(N)}_U$ one has
\begin{equation}
\label{eq:main}
\langle \psi_{U} | |x|_{\infty}^{-1}| \psi_U \rangle \geq R^{-1}
\quad\text{and}\quad
\langle \psi_{U} | \chi_{\{|x|_{\infty}<R \}} | \psi_U \rangle \geq m \,.
\end{equation}
\end{thm}

This theorem immediately implies Theorem \ref{binding}. Indeed, let $U_c\in\mathcal U^{(N)}_c$, let $K$ be the largest integer such that $E_{U_c}^{(N)}=E_{U_c}^{(N-K)}$ and assume that $U_c>\frac{1}{N-K}$. Then there is a non-empty subset of $\mathcal U^{(N)}$, containing $U_c$ in its closure, with elements satisfying \eqref{eq:mainass} for $$
\delta= \frac12\left(U_c - \frac1{N-K}\right)>0
\quad\text{and}\quad
\theta = \frac12\left(E_{U_c}^{(N-K-1)}- E_{U_c}^{(N-1)}\right)>0 \,.
$$
Then Theorem \ref{main} implies that any sequence $\psi_{U_n}$ of normalized ground states of $H^{(N)}_{U_n}$ with $U_n$ from this subset and with $U_n\to U_c$ satisfies \eqref{eq:main}. Since $\psi_{U_n}$ is bounded in $H^1(\R^{3N})$ it has a weak limit in this space which, according to \eqref{eq:main} is not identically zero. One easily verifies that the weak limit is a ground state of $H^{(N)}_{U_c}$, thus proving Theorem \ref{binding}. This observation reduces the proof of Theorem \ref{binding} to the proof of Theorem \ref{main}.

Theorem \ref{main}, in turn, can be deduced from two lemmas that we state next. The first one, and this is the novelty of this paper, is an operator inequality for $H_U^{(N)}$. It estimates $H^{(N)}_U-E^{(N)}_U$ from below by a potential well which is attractive of order $l^{-2}$ for $|x|_\infty<l$, but has a repulsive Coulomb tail for $|x|_\infty>l$. The parameter $l$ can be chosen arbitrary large and the constants are uniform in $U$. The precise statement is

\begin{lem}[Operator inequality]\label{opineq}
For any $\delta,\theta>0$ and $1\leq K\leq N-1$ there are three constants $c,C,l_0>0$ such that for all $U\in\mathcal U^{(N)}$ satisfying \eqref{eq:mainass} and all $l\geq l_0$ one has
\begin{equation}\label{operineq}
H^{(N)}_U-E^{(N)}_U\geq -\frac{C}{l^2} \chi_{\{|x|_{\infty}<l\}} + \left( E_U^{(N-1)}-E^{(N)}_U + \frac{c}{|x|_{\infty}}\right) \chi_{\{|x|_{\infty}\geq l\}} \,.
\end{equation}
\end{lem}

The second ingredient in the proof of Theorem \ref{main} is the following elementary lemma, which we cite from \cite{FLS}.

\begin{lem}[Calculus lemma]\label{calc}
Let $\rho\in L^1(\R^+)$ be non-negative with $\int_0^{\infty} \rho(r)dr =1$. Assume that there are constants $b>0$ and $l_0\geq 0$ such that
\begin{equation}
\frac{b}{l^2}\int_0^l \rho(r)dr\geq \int_l^{\infty} \frac{\rho(r)}{r}dr
\end{equation}
for all $l\geq l_0$. Then 
\begin{equation}
\int_0^{\infty}  \frac{\rho(r)}{r}dr\geq \frac{1}{2(b+l_0)}
\end{equation}
and
\begin{equation}
\int_0^{l_0} \rho(r)dr\geq \frac{l_0^2}{(b+l_0)^2} \,.
\end{equation}
\end{lem}

We now use Lemmas \ref{opineq} and \ref{calc} to give the

\begin{proof}[Proof of Theorem \ref{main}]
Fix $\delta,\theta>0$ and $1\leq K\leq N-1$ and consider $U$'s as in the statement of the theorem. Any normalized ground state $\psi_U$ of $H^{(N)}_U$ satisfies, according to the operator inequality \ref{operineq},
$$
\frac{C}{l^2} \int_{\{|x|_\infty<l\}} |\psi_U|^2 \,dx
\geq c \int_{\{|x|_\infty\geq l\}} \frac{|\psi_U|^2}{|x|_\infty} \,dx
$$
for all $l\geq l_0$ and some constants $c,C,l_0>0$. (Here we also used the fact that $E_U^{(N-1)}\geq E^{(N)}_U$.) Lemma \ref{calc} now implies that
$$
\int_{\R^{3N}} \frac{|\psi_U|^2}{|x|_\infty} \,dx \geq \frac{c}{2(C + l_0 c)}
\quad\text{and}\quad
\int_{\{|x|_\infty<l_0\}} |\psi_U|^2 \,dx \geq \frac{c^2 l_0^2}{(C + l_0 c)^2} \,.
$$
This proves \eqref{eq:main}.
\end{proof}

To summarize the content of this section, we have reduced the proof of our main result to the proof of the operator inequality in Lemma \ref{opineq}. This will be accomplished in Section \ref{sec:opineq} after some preparations in Section \ref{sec:ims}.


\section{A partition of unity}\label{sec:ims}

Our goal in this section is to construct a partition of unity in $\R^{3N}$ with an effective control on the Coulomb interaction between (subsets of) the `particles' (electrons) $x_1,\ldots,x_N\in\R^3$ and another `particle' (nucleus) at the origin. Throughout this section we fix an integer $N\geq 2$.

In order to formulate the properties of our partition of unity we need to introduce some notation. For every integer $1\leq k\leq N$ we denote by $\mathcal J_k$ the collection of all integer sequences $(J_1,\ldots,J_k)$, where $1\leq J_l\leq N$ for all $l$ and where all the $J_l$'s are mutually distinct. For given $J\in\mathcal J_k$ and $x\in\R^{3N}$ we denote by $\hat x_J$ the vector in $\R^{3(N-k)}$ which coincides with $x$, but where the entries $x_{J_1},\ldots,x_{J_k}$ have been erased. Thus, if we denote as usual $|x|_{\infty}:=\max\{|x_1|,..., |x_N|\}$, then
$$
|\hat x_J|_\infty = \max\{ |x_j|: j \neq J_1,\ldots,J_k \} \,.
$$
One can think of $|\hat x_J|_\infty$ as the maximum distance of the $N-k$ electrons from the nucleus after having removed the $k$ electrons with indices in $J$.

Finally, if $\pi\in\mathcal S_N$ is a permutation and $x\in\R^{3N}, J\in\mathcal J_k$, we set
$$
x_\pi = (x_{\pi(1)},\ldots,x_{\pi(N)})
\quad\text{and}\quad \pi(J) = (\pi(J_1),\ldots,\pi(J_k)) \,.
$$

Here is the description of our partition of unity.

\begin{prop}[Partition of unity] \label{partition}
For any integer $1\leq K\leq N-1$ and any $0<\epsilon<1/2$ there is a constant $c_\epsilon>0$ with the following property. For any $l>0$ there is a quadratic partition of unity, 
$$
\Lambda_0(x)^2 + \sum_{j=1}^N \Lambda_j(x)^2 + \sum_{J\in\mathcal J_{K+1}} \Lambda_{J} (x)^2=1
\qquad\text{for all}\ x\in \R^{3N} \,,
$$
with
$$
\Lambda_0(x)= 0 \ \text{unless} \ |x|_{\infty}\leq l
$$
and, for $1\leq j\leq N$,
\begin{align*}
\Lambda_j(x)= 0 \ \text{unless}\ \frac{1}{N-K}\sum_{m\neq j}\frac{1}{|x_m-x_{j}|} \geq\frac{1-2\epsilon}{|x_{j}|}
\end{align*}
and, for $J=(J_1,\ldots,J_{K+1})\in\mathcal J_{K+1}$,
\begin{align*}
\Lambda_J(x)= 0 \ \text{unless}\ & |x|_{\infty} \geq \frac{l}{2}\,, \\
& |x_{J_n}| \geq \frac{1}{2} |\hat x_{J_1,J_2,\ldots,J_{n-1}}|_\infty
\quad \text{for all}\ 1\leq n\leq K+1 \,, \\
& \frac{1}{N-n}\sum_{m\neq J_1,\ldots,J_{n}}\frac{1}{|x_m-x_{J_n}|} \leq\frac{1-\epsilon}{|x_{J_{n}}|} 
\quad \text{for all}\ 1\leq n\leq K \,.
\end{align*}
Moreover, we have
$$
|\nabla\Lambda_0(x)|^2 + \sum_{j=1}^{N} |\nabla\Lambda_j (x)|^2 + \sum_{J\in\mathcal J_{K+1}} |\nabla\Lambda_{J} (x)|^2
\leq \frac{c_{\epsilon}}{l^2}
\qquad\text{if}\ \Lambda_{0}(x)>0
$$
and
$$
|\nabla\Lambda_0(x)|^2 + \sum_{j=1}^{N} |\nabla\Lambda_j (x)|^2 + \sum_{J\in\mathcal J_{K+1}} |\nabla\Lambda_{J} (x)|^2
 \leq \frac{c_{\epsilon}}{l |x|_\infty}
\qquad\text{if}\ \Lambda_{0}(x)<1 \,.
$$
Finally, the $\Lambda$'s behave as follows under permutations $\pi\in\mathcal S_N$.
\begin{align*}
& & \Lambda_0(x_\pi)=\Lambda_0(x) & \,, \\
& \text{for}\ 1\leq j\leq N : & \Lambda_j(x_\pi)=\Lambda_j(x) & \ \text{if}\ \pi(j)=j \,,\\
& \text{for}\ J\in\mathcal J_{K+1}: & \Lambda_J(x_\pi)=\Lambda_J(x) & \ \text{if}\ \pi(J)=J \,.
\end{align*}
\end{prop}

Here, for $n=1$ we set $\hat x_{J_1,J_2,\ldots,J_{n-1}}=x$, that is, the condition $|x_{J_n}| \geq \frac{1}{2} |\hat x_{J_1,J_2,\ldots,J_{n-1}}|_\infty$ for $n=1$ means $|x_{J_1}| \geq \frac{1}{2} |x|_\infty$.

Before proving this proposition we record some useful geometric facts.

\begin{lem}\label{imp2}
Let $1\leq n\leq N-1$, $(J_1,\ldots,J_n)\in \mathcal J_n$ and assume that $x\in\R^{3N}$ satisfies
$$
\frac{1}{N-n}\sum_{m\neq J_1,\ldots,J_{n}}\frac{1}{|x_m-x_{J_n}|} \leq\frac{1-\epsilon}{|x_{J_{n}}|} \,.
$$
Then
\begin{equation}
\label{eq:lem1}
\min_{m\neq J_1,\ldots,J_n} |x_m-x_{J_n}| \geq \frac{|x_{J_n}|}{(N-n)(1-\epsilon)}
\end{equation}
and
\begin{equation}
\label{eq:lem2}
|\hat x_{J_{1},J_{2},..,J_{n}}|_{\infty} \geq \frac{\epsilon}{1-\epsilon}|x_{J_n}| \,.
\end{equation}
\end{lem}

\begin{proof}
To prove the first claim it suffices to notice that
$$
\max_{m\neq J_1,\ldots J_n}\frac{1}{|x_m-x_{J_n}|}\leq \sum_{m\neq J_1,\ldots J_n} \frac{1}{|x_m-x_{J_n}|}\leq \frac{(N-n)(1-\epsilon)}{|x_{J_n}|} \,.
$$

To prove the second claim we shall show that $|\hat x_{J_{1},J_{2},..,J_{n}}|_{\infty} < \frac{\epsilon}{1-\epsilon}|x_{J_n}|$ implies $\frac{1}{N-n}\sum_{m\neq J_1,\ldots,J_{n}}\frac{1}{|x_m-x_{J_n}|} >\frac{1-\epsilon}{|x_{J_{n}}|}$. Thus, assume that for all $m \neq  J_1,\ldots,J_n$ one has $|x_m|<\frac{\epsilon}{1-\epsilon}|x_{J_n}|$. By the triangle inequality we get
$$
|x_m-x_{J_n}|\leq |x_m|+|x_{J_n}|<\left(\frac{\epsilon}{1-\epsilon}+1\right)|x_{J_n}|=\frac{1}{1-\epsilon}|x_{J_n}| \,,
$$
and therefore
$$
\frac{1}{N-n}\sum_{m\neq J_1,\ldots,J_{n}}\frac{1}{|x_m-x_{J_n}|} >\frac{1-\epsilon}{|x_{J_{n}}|} \,,
$$
as claimed.
\end{proof}

The following result is a consequence of the construction of the partition of unity in Proposition \ref{partition} and of Lemma \ref{imp2}. It says that on the support of $\Lambda_J$ with $J\in\mathcal J_{K+1}$, all $K+1$ particles $x_{J_1},\ldots,x_{J_{K+1}}$ are `far out', that is, at a distance comparable to the distance of the particle that is farthest out. This will be useful in the proof of Lemma \ref{opineq}.

\begin{lem}\label{farout}
Let $1\leq k\leq N$ and $(J_1,\ldots,J_k)\in\mathcal J_k$. Assume that $x\in\R^{3N}$ satisfies
\begin{align*}
& |x_{J_n}| \geq \frac{1}{2} |\hat x_{J_1,J_2,\ldots,J_{n-1}}|_\infty
\quad \text{for all}\ 1\leq n\leq k \,, \\
& \frac{1}{N-n}\sum_{m\neq J_1,\ldots,J_{n}}\frac{1}{|x_m-x_{J_n}|} \leq\frac{1-\epsilon}{|x_{J_{n}}|} 
\quad \text{for all}\ 1\leq n\leq k-1 \,.
\end{align*}
Then
\begin{equation}
\label{eq:farout}
|x_{J_n}| \geq \frac 1{2^n} \left( \frac{\epsilon}{1-\epsilon} \right)^{n-1} |x|_\infty
\end{equation}
for every $1\leq n\leq k$.

In particular, with the notation of Proposition \ref{partition}, let $J\in\mathcal J_{K+1}$ and let $x\in\R^{3N}$ such that $\Lambda_J(x)\neq 0$. Then \eqref{eq:farout} holds for every $1\leq n\leq K+1$.
\end{lem}

\begin{proof}
By \eqref{eq:lem2} the second assumption yields $|\hat x_{J_{1},J_{2},..,J_{n}}|_{\infty} \geq \frac{\epsilon}{1-\epsilon}|x_{J_n}|$ for all $1\leq n\leq k-1$. Combining this with the first assumption, we get
$$
|x_{J_n}| \geq \frac12 |\hat x_{J_{1},J_{2},..,J_{n-1}}|_{\infty} \geq \frac{\epsilon}{2(1-\epsilon)}|x_{J_{n-1}}|
\qquad\text{for all}\ 2\leq n\leq k \,.
$$
Iterating this bound and recalling that $|x_{J_1}|\geq \frac12|x|_\infty$ we obtain claimed inequality.

If $\Lambda_J(x)\neq 0$ for some $J\in\mathcal J_{K+1}$, then the above assumptions are satisfied by the construction of $\Lambda_J$.
\end{proof}

We now turn to the

\begin{proof}[Proof of Proposition \ref{partition}]
We begin by constructing a finer partition of unity,
$$
\Lambda_0(x)^2 + \sum_{k=1}^{K+1} \sum_{J\in\mathcal J_k} \Lambda_{J} (x)^2=1
\qquad\text{for all}\ x\in \R^{3N} \,.
$$
The terms $\Lambda_0$ and $\Lambda_J$ with $J\in\mathcal J_{K+1}$ are as in the statement of the proposition. For $J=(J_1,\ldots,J_k)\in\mathcal J_k$ with $1\leq k\leq K$, we will have
\begin{align*}
\Lambda_J(x)= 0 \ \text{unless}\ & |x|_{\infty} \geq \frac{l}{2}\,, \\
& |x_{J_n}| \geq \frac{1}{2} |\hat x_{J_1,J_2,\ldots,J_{n-1}}|_\infty
\quad \text{for all}\ 1\leq n\leq k \,, \\
& \frac{1}{N-n}\sum_{m\neq J_1,\ldots,J_{n}}\frac{1}{|x_m-x_{J_n}|} \leq\frac{1-\epsilon}{|x_{J_{n}}|} 
\quad \text{for all}\ 1\leq n\leq k-1 \,, \\
& \frac{1}{N-k}\sum_{m\neq J_1,\ldots,J_{k}}\frac{1}{|x_m-x_{J_k}|} \geq\frac{1-2\epsilon}{|x_{J_k}|} \,.
\end{align*}
We will also prove that
\begin{equation}
\label{eq:grad1}
|\nabla\Lambda_0(x)|^2 + \sum_{k=1}^{K+1} \sum_{J\in\mathcal J_k} |\nabla\Lambda_{J} (x)|^2 \leq \frac{c_{\epsilon}}{l^2}
\qquad\text{if}\ \Lambda_{0}(x)>0
\end{equation}
and
\begin{equation}
\label{eq:grad2}
|\nabla\Lambda_0(x)|^2 + \sum_{k=1}^{K+1} \sum_{J\in\mathcal J_k}  |\nabla\Lambda_{J} (x)|^2 \leq \frac{c_{\epsilon}}{l |x|_\infty}
\qquad\text{if}\ \Lambda_{0}(x)<1 \,.
\end{equation}
Moreover, the refined partition of unity has the symmetry property
\begin{equation}
\label{eq:symm}
\Lambda_{J}(x_\pi)=\Lambda_{\pi(J)}(x)
\end{equation}
for all $J\in\mathcal J_{k}$, $1\leq k\leq K+1$, and $\Lambda_0(x_\pi)=\Lambda_0(x)$.

We now explain how to obtain the asserted partition of unity from the refined one. We keep the terms $\Lambda_0$ and $\Lambda_J$ with $J\in\mathcal J_{K+1}$. Moreover, for $1\leq j\leq N$ we define
$$
\Lambda_j(x) = \sqrt{\sum_{k=1}^K \sum_{J\in \mathcal J_k,\, J_k=j} \Lambda_J(x)^2 }
$$
Thus, if $\Lambda_j(x)\neq 0$, then there is a $1\leq k\leq K$ and a $J\in\mathcal J_k$ with $J_k = j$ such that
$$
\frac{1}{N-k}\sum_{m\neq J_1,\ldots,J_{k}}\frac{1}{|x_m-x_{j}|} \geq\frac{1-2\epsilon}{|x_{j}|} \,.
$$
But this inequality implies that also
$$
\frac{1}{N-K}\sum_{m\neq j}\frac{1}{|x_m-x_{j}|} \geq\frac{1-2\epsilon}{|x_{j}|} \,,
$$
which is the asserted condition for $\Lambda_j$. Moreover, by the Schwarz inequality,
$$
|\nabla \Lambda_j| \leq \sqrt{\sum_{k=1}^K \sum_{J\in \mathcal J_k,\, J_k=j} |\nabla \Lambda_J|^2 } \,.
$$
Therefore, the gradient bounds for the coarser partition of unity follow from those for the refined one.

The reason why we pass from the refined partition to the one stated in the proposition is that the latter satisfies the required symmetry assumptions. This is an immediate consequence of \eqref{eq:symm}.

Thus, it remains to construct the refined partition of unity. Let us start by introducing the following functions $f,g:\R^+\rightarrow\R^+$,
\begin{align*}      
   f(s) & =
   \begin{cases}
   1 & \text{ if } 0< s\leq \frac{1}{2} \,, \\
   \cos(\pi (s-\frac{1}{2})) & \text{ if } \frac{1}{2}< s\leq 1 \,, \\
   0 & \text{ if } s\geq 1 \,,
   \end{cases} \\
   g(s) & =
   \begin{cases}
   0 & \text{ if } 0< s\leq 1-2\epsilon \,, \\
   \sin(\pi \frac{(s-(1-2\epsilon))}{2\epsilon}) & \text{ if  } 1-2\epsilon<s<1-\epsilon \,, \\
   1 & \text{ if }s\geq 1-\epsilon \,,
    \end{cases}
\end{align*}
and the corresponding $\tilde f,\tilde g:\R^+\rightarrow\R^+$ fulfilling  $f^2(s)+\tilde f^2(s)=g^2(s)+\tilde g^2(s)=1$ for all $s\geq 0$. With this notation we define
$$
\Lambda_0(x)=f\left(\frac{|x|_{\infty}}{l}\right) \,,
$$
For $J=(J_1,\ldots,J_k)\in\mathcal J_k$ with $1\leq k\leq K$ we set
\begin{align*}
\Lambda_J(x) = & \tilde f \left( \frac{|x|_{\infty}}{l} \right) \\
& \times \prod_{n=1}^{k-1} 
\left( 
\frac{\tilde f\left(\frac{|x_{J_n}|}{|\hat x_{J_1,\ldots,J_{n-1}}|_\infty}\right)}{\sqrt{ \sum_{j\neq J_1,\ldots,J_{n-1}} \tilde f\left(\frac{|x_{j}|}{|\hat x_{J_1,\ldots,J_{n-1}}|_\infty}\right)^2}}\ 
\tilde g\left( \frac{1}{N-n} \sum_{m\neq J_1,\ldots,J_n} \frac{|x_{J_n}|}{|x_m-x_{J_n}|} \right) 
\right) \\
& \times 
\frac{\tilde f\left(\frac{|x_{J_k}|}{|\hat x_{J_1,\ldots,J_{k-1}}|_\infty}\right)}{\sqrt{ \sum_{j\neq J_1,\ldots,J_{k-1}} \tilde f\left(\frac{|x_{j}|}{|\hat x_{J_1,\ldots,J_{k-1}}|_\infty}\right)^2}}\  
g\left( \frac{1}{N-k} \sum_{m\neq J_1,\ldots,J_k} \frac{|x_{J_k}|}{|x_m-x_{J_k}|} \right) \,.
\end{align*}
For $J=(J_1,\ldots,J_{K+1})\in\mathcal J_{K+1}$ we set
\begin{align*}
\Lambda_J(x) = & \tilde f \left( \frac{|x|_{\infty}}{l} \right) \\
& \times \prod_{n=1}^K 
\left( 
\frac{\tilde f\left(\frac{|x_{J_n}|}{|\hat x_{J_1,\ldots,J_{n-1}}|_\infty}\right)}{\sqrt{ \sum_{j\neq J_1,\ldots,J_{n-1}} \tilde f\left(\frac{|x_{j}|}{|\hat x_{J_1,\ldots,J_{n-1}}|_\infty}\right)^2}}\  
\tilde g\left( \frac{1}{N-n} \sum_{m\neq J_1,\ldots,J_n} \frac{|x_{J_n}|}{|x_m-x_{J_n}|} 
\right) \right) \\
& \times \frac{\tilde f\left(\frac{|x_{J_{K+1}}|}{|\hat x_{J_1,\ldots,J_{K}}|_\infty}\right)}{\sqrt{ \sum_{j\neq J_1,\ldots,J_{K}} \tilde f\left(\frac{|x_{j}|}{|\hat x_{J_1,\ldots,J_{K}}|_\infty}\right)^2}} \,.
\end{align*}

It is straightforward but somewhat tedious to verify that this defines indeed a quadratic partition of unity. One way to verify this is to begin to sum $\Lambda_J(x)^2$ over all $J\in \mathcal J_{K+1}$ whose $K$ first entries $(J_1,\ldots,J_K)$ agree with some given element in $\mathcal J_K$. After having performed this sum, we sum $\Lambda_J(x)^2$ over all $J\in \mathcal J_{K+1}$ and all $J\in\mathcal J_K$ whose $K-1$ first entries $(J_1,\ldots,J_{K-1})$ agree with some given element in $\mathcal J_{K-1}$. This sum can be simplified using $g^2 + \tilde g^2=1$. Proceeding in this way, we obtain the claimed partition of unity property.

The claimed support conditions and the symmetry condition \eqref{eq:symm} follow immediately. We also observe that the terms in the denominators are positive, since for any $x\in\R^{3N}$ and any $(J_1,\ldots,J_{n-1})\in\mathcal J_{n-1}$ there is a $j\neq J_1,\ldots,J_{n-1}$ such that $|x_j|=|\hat x_{J_1,\ldots,J_{n-1}}|_\infty$ and therefore, since $\tilde f(1)=1$,
$$
\sum_{j\neq J_1,\ldots,J_{n-1}} \tilde f\left(\frac{|x_{j}|}{|\hat x_{J_1,\ldots,J_{n-1}}|_\infty}\right)^2 \geq 1 \,.
$$

Now we prove the gradient estimates \eqref{eq:grad1} and \eqref{eq:grad2}. The functions $f$, $\tilde f$, $g$, $\tilde g$ and $|\cdot|_\infty$ are Lipschitz, and therefore $\Lambda_J$ is so as well and it suffices to prove the gradient bounds only in subsets where the numbers $|x_1|, \ldots,|x_N|$ are pairwise distinct. Moreover, since all factors in the definition of $\Lambda_J$ are bounded, it suffices to prove that each factor individually satisfies the claimed gradient bounds.

We begin with the term $\tilde f(|x|_\infty/l)$. Differentiating we get
\begin{equation*}      
 \nabla_i \left( \tilde{f}\left(\frac{|x|_{\infty}}{l}\right) \right) = 
 \begin{cases}
 \frac1l \tilde f'(|x|_\infty/l) \frac{x_i}{|x_i|}
 & \text{ if } |x|_{\infty} = |x_i| \,,\\
 0
 & \text{ if } |x|_{\infty}\neq |x_i| \,.
 \end{cases}
\end{equation*}
This is obviously bounded by
$$
\left| \nabla_i \left( \tilde{f}\left(\frac{|x|_{\infty}}{l}\right) \right) \right| \leq \frac{\pi}{l} \,.
$$
Moreover, since $\tilde f'$ is supported in $[1/2,1]$, we can bound $|x|_\infty\leq l$ on the support of $\nabla_i \left( \tilde{f}\left(\frac{|x|_{\infty}}{l}\right)\right)$ and find
$$
\left| \nabla_i \left( \tilde{f}\left(\frac{|x|_{\infty}}{l}\right) \right) \right| \leq \frac{\pi}{\sqrt{l |x|_\infty}} \,.
$$
The argument for the term $f(|x|_\infty/l)$ is the same.

Analogously, we have
\begin{equation*}      
 \nabla_i \left( \tilde{f}\left(\frac{|x_{J_n}|}{|\hat x_{J_1,\ldots,J_{n-1}}|_{\infty}}\right) \right) = 
 \begin{cases}
 \frac1{|\hat x_{J_1,\ldots,J_{n-1}}|_{\infty}} \tilde f'(\cdot) \frac{x_i}{|x_i|}
 & \text{ if } i=J_n \text{ and } |\hat x_{J_1,\ldots,J_{n-1}}|_{\infty} \neq |x_i| \,,\\
 - \frac{|x_{J_n}|}{|\hat x_{J_1,\ldots,J_{n-1}}|_{\infty}^2}
 \tilde f'(\cdot) \frac{x_i}{|x_i|}
  & \text{ if } i\neq J_n \text{ and } |\hat x_{J_1,\ldots,J_{n-1}}|_{\infty} = |x_i| \,,\\
 0
 & \text{ if } i\neq J_n \text{ and } |\hat x_{J_1,\ldots,J_{n-1}}|_{\infty} \neq |x_i| \,.
 \end{cases}
\end{equation*}
Thus,
$$
\left| \nabla_i \left( \tilde{f}\left(\frac{|x_{J_n}|}{|\hat x_{J_1,\ldots,J_{n-1}}|_{\infty}}\right) \right) \right| \leq \frac{\pi}{|\hat x_{J_1,\ldots,J_{n-1}}|_{\infty}} \leq \frac{\pi}{|x_{J_n}|} \,.
$$
On the support of $\Lambda_J$ (with $J\in\mathcal J_k$ for some $1\leq k\leq K+1$) the assumptions of Lemma \ref{farout} are satisfied and we can bound $|x_{J_n}|^{-1}$ from above in terms of $|x|_\infty^{-1}$. Thus,
\begin{align*}
\left| \nabla_i \left( \tilde{f}\left(\frac{|x_{J_n}|}{|\hat x_{J_1,\ldots,J_{n-1}}|_{\infty}}\right) \right) \right| \leq \frac{c_\epsilon'}{|x_{J_n}|}
\,.
\end{align*}
Because of the factor $\tilde f(|x|_\infty/l)$ contained in $\Lambda_J$, we may assume that $|x|_\infty \geq l/2$. This allows us to replace the term $|x|_\infty^{-1}$ in the above bound by either $2 l^{-1}$ or $(2/l|x|_\infty)^{1/2}$. This yields the claimed bound.

Finally, we compute
\begin{align*}
& \nabla_i \left( \tilde g\left( \frac{1}{N-n} \sum_{m\neq J_1,\ldots,J_n} \frac{|x_{J_n}|}{|x_m-x_{J_n}|} \right) \right) \\
& = \begin{cases}
\frac{1}{N-n} \tilde g^{'}(\cdot) \left( 
\sum_{m\neq J_1,\ldots,J_n} \frac{1}{|x_m-x_i|} \frac{x_i}{|x_i|}
- \sum_{m\neq J_1,\ldots,J_n} \frac{|x_{i}|}{|x_m-x_{i}|^2} \frac{x_i-x_{m}}{|x_i-x_{m}|^2} \right)
& \text{ if }  i = J_n \,, \\
- \frac{1}{N-n} \tilde g^{'}(\cdot) \frac{|x_{J_n}|}{|x_i-x_{J_n}|^2} \frac{x_i-x_{J_n}}{|x_i-x_{J_n}|}
& \text{ if }  i \neq J_n \,,
\end{cases}
\end{align*}
and find
\begin{align*}
& \left| \nabla_i \left( \tilde g\left( \frac{1}{N-n} \sum_{m\neq J_1,\ldots,J_n} \frac{|x_{J_n}|}{|x_m-x_{J_n}|} \right) \right) \right| \\
& \leq \frac{1}{N-n} \frac{\pi}{2\epsilon} 
\sum_{m\neq J_1,\ldots,J_n} \left( \frac{1}{|x_m-x_{J_n}|}
+ \frac{|x_{J_n}|}{|x_m-x_{J_n}|^2} \right)
\end{align*}
Since $\tilde g'$ has support in $[1-2\epsilon, 1-\epsilon]$, it suffices to bound the gradient on the set where
$$
1-2\epsilon\leq \frac{1}{N-n} \sum_{m\neq J_1,\ldots,J_n} \frac{|x_{J_n}|}{|x_m-x_{J_n}|} \leq 1-\epsilon \,.
$$
By \eqref{eq:lem1} this bound entails $|x_m-x_{J_n}| \geq |x_{J_n}|/ ((N-n)(1-\epsilon))$ for all $m\neq J_1,\ldots, J_m$, and therefore
\begin{align*}
\sum_{m\neq J_1,\ldots,J_n} \!\!\left( \frac{1}{|x_m-x_{J_n}|}
+ \frac{|x_{J_n}|}{|x_m-x_{J_n}|^2} \right)
& \leq \left( 1 + (N-n)(1-\epsilon) \right) \!\!\sum_{m\neq J_1,\ldots,J_n}\! \frac{1}{|x_m-x_{J_n}|} \\
& \leq \left( 1 + (N-n)(1-\epsilon) \right) (1-\epsilon) \frac{1}{|x_{J_n}|} \,.
\end{align*}
We can now argue as before and use Lemma \ref{farout} to bound $|x_{J_n}|^{-1}$ from above in terms of $|x|_\infty^{-1}$. Then, again because of the factor $\tilde f(|x|_\infty/l)$, we may replace $|x|_\infty^{-1}$ in the above bound by either $2 l^{-1}$ or $(2/l|x|_\infty)^{1/2}$. This yields the claimed bound. The proof for $g$ instead of $\tilde g$ is similar and is omitted.
\end{proof}

\begin{remark}\label{symm}
As already explained in the proof, the argument of passing from the refined partition of unity to the one stated in Proposition \ref{partition} is needed to obtain the symmetry properties. Those are needed since in our main theorem \ref{binding} we consider the operator $H^{(N)}_U$ on \emph{anti-symmetric} functions. If, instead, we had considered $H^{(N)}_U$ without any symmetry restrictions, the refined partition would have been sufficient for the proof of our results.
\end{remark}


\section{Proof of Lemma \ref{opineq}}\label{sec:opineq}

With the help of the partition of unity that we constructed in the previous section we are now able to give the

\begin{proof}[Proof of  Lemma \ref{opineq}]
Let $1\leq K\leq N-1$, let $\delta,\theta>0$ and assume that $U$ satisfies \eqref{eq:mainass}. Our argument makes use of two additional parameters $\epsilon>0$ and $l>0$ that we will specify later depending on $\delta$ and $\theta$.

We use the partition of unity from Proposition \ref{partition} (with the given parameters $K$, $\epsilon$ and $l$) and the IMS formula (see, e.g., \cite{CyFrKiSi}) to localize the Hamiltonian. That is, we write for any wave function $\psi$
$$
\langle \psi | H_U^{(N)} | \psi \rangle= \sum_{j=0}^N e_j \|\psi_{j}\|^2 + \sum_{J\in\mathcal J_{K+1}} e_J \|\psi_J\|^2
$$
where $\psi_j=\Lambda_j\psi$ and $\psi_J=\Lambda_J\psi$. Here
$$
e_j = \frac{\langle \psi_j | H_U^{(N)} - \sum_{j=0}^N |\nabla\Lambda_j|^2 - \sum_{J\in\mathcal J_{K+1}} |\nabla\Lambda_{J}|^2 | \psi_j \rangle}{\|\psi_{j}\|^2}
$$
and similarly for $e_J$. Our goal is to show lower bounds on $e_j$ and $e_J$.

For $e_0$, namely on the support of $\Lambda_0$, we know from Proposition \ref{partition} that the localization error is bounded by $c_\epsilon/l^2$. Moreover, since $\Lambda_0$ is a symmetric function, $\psi_0$ is anti-symmetric. Thus, bounding $H_U^{(N)}$ from below by $E_U^{(N)}$ we immediately arrive at
$$
e_0 \geq E_U^{(N)} -\frac{c_{\epsilon}}{l^2} \,.
$$

Now let $1\leq j\leq N$. On the support of $\Lambda_j$ the localization error is bounded by $c_\epsilon/(l|x|_\infty)$. Moreover, we split
$$
H_U^{(N)} = H_U^{(N-1)} + p_{j}^2 - \frac{1}{|x_{j}|} + \sum_{m\neq j} \frac{U}{|x_m-x_{j}|} \,,
$$
where $H_U^{(N-1)}$ is obtained from $H_U^{(N)}$ by dropping all terms involving the coordinate $x_{j}$. By construction of $\Lambda_j$, $\psi_j$ is an anti-symmetric function of the variables $\hat x_j$. Thus,
$$
H_U^{(N)} \geq E_U^{(N-1)} - \frac{1}{|x_{j}|} + \sum_{m\neq j} \frac{U}{|x_m-x_{j}|} \,.
$$
If we now use the fact that $\psi_j$ has support where $\Lambda_j\neq 0$, we can further bound
\begin{align*}
e_j & \geq E_U^{(N-1)} + \frac{1}{\|\psi_{j}\|^2} \langle \psi_{j} | \left( - \frac{1}{|x_{j}|} + \sum_{m\neq j} \frac{U}{|x_m-x_{j}|} - \frac{c_{\epsilon}l^{-1}}{|x|_{\infty}} \right) |\psi_{j}\rangle \\
& \geq E_U^{(N-1)} + \frac{1}{\|\psi_{j}\|^2} \langle \psi_{j} | \left( \frac{U(1-2\epsilon)(N-K) - 1}{|x_{j}|} - \frac{c_{\epsilon}l^{-1}}{| x|_{\infty}} \right)|\psi_{j} \rangle \,.
\end{align*}
We now choose $\epsilon>0$ so small that
$$
(N-K)\delta - 2\epsilon \left(1+(N-K)\delta \right) \geq \frac{\delta}{2} \,.
$$
Because of the first assumption on $U$ in \eqref{eq:mainass}, this implies that
$$
U(1-2\epsilon)(N-K) - 1 
\geq \left( \frac{1}{N-K} + \delta \right) (1-2\epsilon) (N-K) - 1  
\geq \frac{\delta}{2} >0 \,,
$$
and therefore
\begin{align*}
e_j \geq E_U^{(N-1)} + \frac{1}{\|\psi_{j}\|^2} \langle \psi_{j} | \frac{\delta/2 - c_\epsilon/l}{|x|_\infty} |\psi_{j} \rangle \,.
\end{align*}

Finally, we bound $e_J$ with $J=(J_1,\ldots,J_{K+1})\in \mathcal J_{K+1}$. As before, on the support of $\Lambda_J$ the localization error is bounded by $c_\epsilon/(l|x|_\infty)$. By construction of $\Lambda_{J}$, $\psi_J$ is an anti-symmetric function of the variables $\hat x_{J_1,\ldots,J_{K+1}}$. By splitting off all the coordinates $x_{J_n}$ with $1\leq n\leq K+1$ we obtain the lower bound
$$
H_U^{(N)} \geq E_U^{(N-K-1)} - \sum_{n=1}^{K+1} \frac{1}{|x_{J_n}|} \,.
$$
On the support of $\Lambda_J$ we can use Lemma \ref{farout} to control the last sum and we obtain
\begin{align*}
e_J & \geq E_U^{(N-K-1)} - \frac{1}{\|\psi_{J}\|^2} \langle \psi_{J} | \left( \sum_{n=1}^{K+1} \frac{1}{|x_{j_n}|} + \frac{c_{\epsilon}l^{-1}}{|x|_{\infty}} \right) |\psi_{J}\rangle \\
& \geq E_U^{(N-K-1)} - \frac{1}{\|\psi_{J}\|^2} \langle \psi_{J} | \frac{A_\epsilon+ c_{\epsilon}l^{-1}}{| x|_{\infty}} |\psi_{J} \rangle
\end{align*}
where $A_\epsilon = \sum_{n=1}^{K+1} 2^n \left(1/\epsilon- 1\right)^{n-1}$. Notice that, since $|x|_\infty\geq l/2$ on the support of $\Lambda_J$, we have there
$$
-\frac{A_\epsilon}{|x|_{\infty}} \geq -(2A_\epsilon + \delta)l^{-1}+ \frac{\delta/2}{|x|_\infty}
$$
and therefore
\begin{align*}
e_J & \geq E_U^{(N-K-1)} - (2A_\epsilon + \delta)l^{-1} - \frac{1}{\|\psi_{J}\|^2} \langle \psi_{J} | \frac{\delta/2 + c_{\epsilon}l^{-1}}{| x|_{\infty}} |\psi_{J} \rangle \,.
\end{align*}
At this point, we choose $l_0>0$ such that
$$
\delta/2 + c_\epsilon/l_0 \geq \delta/4
\quad\text{and}\quad
(2A_\epsilon + \delta)l_0^{-1} \leq \theta \,.
$$
Then, by the second assumption on $U$ in \eqref{eq:mainass},
$$
E_U^{(N-K-1)} - (2A_\epsilon + \delta)l_0^{-1} \geq E_U^{(N-1)}
$$
and therefore
$$
e_J \geq E_U^{(N-1)} - \frac{1}{\|\psi_{J}\|^2} \langle \psi_{J} | \frac{\delta/4}{| x|_{\infty}} |\psi_{J} \rangle \,.
$$
for all $l\geq l_0$.

To summarize, we have shown that for all $1\leq j\leq N$ and all $J\in\mathcal J_{K+1}$ we have
$$
e_j \geq E_U^{(N-1)} + \frac{1}{\|\psi_{J}\|^2} \langle \psi_{J} | \frac{\delta/4}{|x|_\infty} |\psi_{J} \rangle
\quad\text{and}\quad
e_J \geq E_U^{(N-1)} + \frac{1}{\|\psi_{J}\|^2} \langle \psi_{J} | \frac{\delta/4}{|x|_\infty} |\psi_{J} \rangle \,.
$$
for all $l\geq l_0$. Moreover, recall that
$$
e_0 \geq E_U^{(N)} -\frac{c_{\epsilon}}{l^2} \,.
$$
Thus, using the fact that
$$
\Lambda_0(x)^2\leq \theta(l-|x|_{\infty}) \,,
\qquad
1-\Lambda_0(x)^2\geq \theta(|x|_{\infty}-l) \,,
$$
we obtain the claimed operator inequality \eqref{operineq}.
\end{proof}



\begin{thebibliography}{99}

\bibitem{Be} H. A. Bethe, \textit{Berechnung der Elektronenaffinit\"at
des Wasserstoffs}, Z. Physik \textbf{57} (1929), 815--821.

\bibitem{CyFrKiSi} H. Cycon, R. Froese, W. Kirsch, B. Simon, \textit{Schr\"odinger operators with application to quantum mechanics and global geometry}. Springer Verlag, 1987.

\bibitem{FLS} R. L. Frank, E. H. Lieb, R. Seiringer, \emph{Binding of polarons and atoms at threshold}, Comm. Math. Phys. \textbf{313} (2012), no. 2, 405--424.

\bibitem{G} D. K. Gridnev, \textit{Bound states at threshold resulting from Coulomb repulsion}. J. Math. Phys. \textbf{53} (2012), 102108.

\bibitem{HOS} M. Hoffmann-Ostenhof,  M. Hoffmann-Ostenhof, B. Simon, \emph{A multiparticle Coulomb system with bound state at threshold}, J. Phys. A \textbf{16} (1983), 1125--1131.

\bibitem{HO} M. Hoffmann-Ostenhof,  M. Hoffmann-Ostenhof, \emph{Absence of an $L^2$ eigenfunction at the bottom of the spectrum of the hydrogen negative ion in the triplet $S$-sector}. J. Phys. A \textbf{17} (1984), 3321--3325.

\bibitem{Li} E. H. Lieb, \emph{Bound on the maximum negative ionization of atoms and molecules}. Phys. Rev. A 29 (1984), 3018--3028.

\bibitem{LiSiSiTh} E. H. Lieb, I. M. Sigal, B. Simon, W. Thirring, \textit{Asymptotic neutrality of large-$Z$ ions}. Comm. Math. Phys. \textbf{116} (1988), 635--644.

\bibitem{Na} P. T. Nam, \textit{New bounds on the maximum ionization of atoms}. Comm. Math. Phys. \textbf{312} (2012), no. 2, 427-–445.

\bibitem{Ru} M. B. Ruskai, \textit{Absence of discrete spectrum in highly negative ions, II. Extension to Fermions}. Comm. Math. Phys. \textbf{82} (1982), 325-–327.

\bibitem{SeSiSo} L. A. Seco, I. M. Sigal, J. P. Solovej, \textit{Bound on the ionization energy of large atoms}. Comm. Math. Phys. \textbf{131} (1990), 307--315.

\bibitem{Si} I. M. Sigal, \textit{Geometric methods in the quantum many-body problem. Nonexistence of very negative ions}. Comm. Math. Phys. \textbf{85} (1982), 309--324.

\bibitem{St} F. H. Stillinger, D. K. Stillinger, \textit{Nonlinear variational study of perturbation theory for atoms and ions}. Phys. Rev. A \textbf{10} (1974), 1109--1124.

\bibitem{Zh} G. Zhislin, \textit{Discussion of the spectrum of Schr\"odinger operator for system of many particles}. Trudy. Mosk. Mat. Ob\v s\v c. \textbf{9} (1960), 81--120.

\end{thebibliography}
\end{document}